%% file: main.tex
\documentclass[12pt]{article}

\usepackage[margin=1.8cm]{geometry}
\usepackage{microtype}
\usepackage{xcolor}
\usepackage{wrapfig}
\input{defs}
\input{theoremsetup}

\makeatletter
\renewcommand\paragraph{\@startsection{paragraph}{4}{\z@}%
  {1ex \@plus .5ex \@minus .2ex}%
  {-1em}%
  {\normalfont\normalsize\bfseries}}
\makeatother

\usepackage[backend=biber, isbn=false, style=alphabetic, backref=true, doi=true, url=false, eprint=false, maxcitenames=10, mincitenames=3, maxbibnames=10, minbibnames=6, minalphanames=3, maxalphanames=5, defernumbers=true]{biblatex}

\usepackage{makecell}
\usepackage{booktabs}
\usepackage{array}
\usepackage{tablefootnote}
\usepackage{titlesec}
\titlespacing*{\section}
  {0pt}    % left spacing
  {1em}    % space before
  {0.5em}  % space after

\makeatletter
\newcommand{\overviewpage}{%
  \clearpage
  \thispagestyle{plain}%
  \noindent\makebox[\textwidth]{%
    \includegraphics[width=\textwidth,height=\textheight,keepaspectratio,clip]{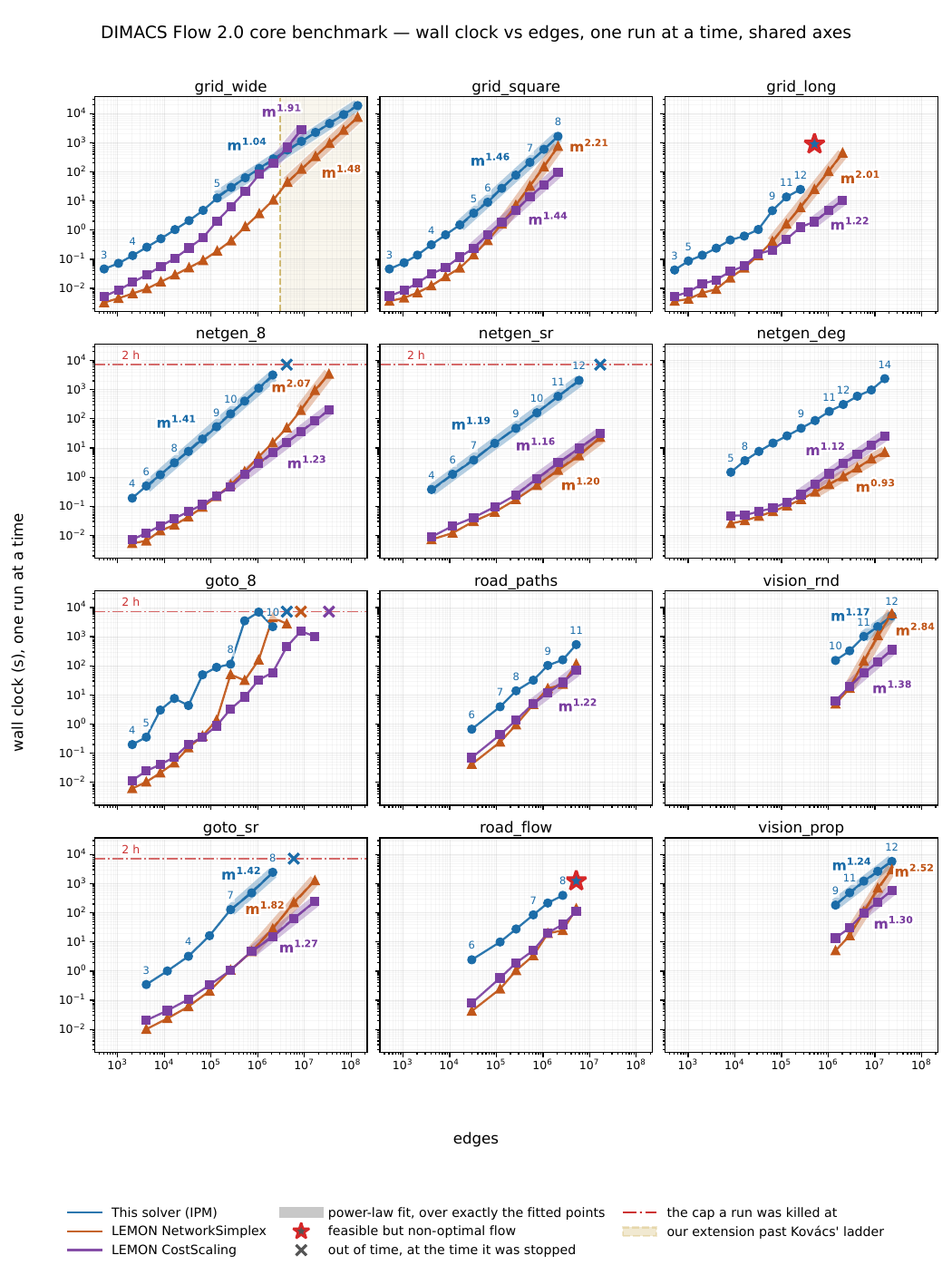}}%
  \edef\@currentlabel{\thepage}\label{page:overview}%
  \clearpage}
\makeatother

\title{A Practical Combinatorial Interior Point Method for Minimum-Cost Flow}
\author{
    Emil Schätzle \\
  \small ETH Zurich \\
  \small \texttt{emil.schaetzle@inf.ethz.ch}
  \and    
   Simon Meierhans\thanks{The research leading to these results has received funding from the starting grant “A New Paradigm for Flow and Cut Algorithms” (no. TMSGI2 218022) and grant no. 200021 204787 of the Swiss National Science Foundation.} \\
  \small ETH Zurich \\
  \small \texttt{mesimon@inf.ethz.ch}
  \and
   Rasmus Kyng\footnotemark[1] \\
  \small ETH Zurich \\
  \small \texttt{kyng@inf.ethz.ch}
  \and
  Maximilian Probst Gutenberg\footnotemark[1] \\
  \small ETH Zurich \\
  \small \texttt{maximilian.probst@inf.ethz.ch}
}

\date{September 19, 2026}
\begin{document}
\maketitle
% First page carries no number; numbering starts on the second page at 1.
\thispagestyle{empty}
\setcounter{page}{0}

\begin{abstract}
In a recent breakthrough, Chen, Kyng, Liu, Peng, Probst Gutenberg, and Sachdeva~\cite{chen2025maximum} designed the first almost-linear-time algorithm for minimum-cost flow. Their approach uses an Interior Point Method (IPM) to achieve complexity $m^{1+o(1)}\log U \log C$. However, the $m^{1+o(1)}$ factor hides intricate recursive data structures, large constants and multiple log factors. In this paper, we give an implementation following the core ideas of their IPM with practical data structures and empirically obtained decision rules. We compare our solver with the cost scaling and network simplex solvers of LEMON on the instance families of the DIMACS Flow Challenge~2.0 core benchmark. Our implementation scales competitively on five of twelve families and better on two. On \texttt{grid\_wide} we observe close to linear scaling: $m^{1.04}$ against $m^{1.48}$ for LEMON. We read these results as first evidence that the almost-linear-time framework can be relevant in practice, discuss our insights gained during the implementation process, and motivate where future work should focus to bridge the theory-practice gap.
\end{abstract}

\clearpage
\section{Introduction}

The minimum-cost flow problem and its special case, the maximum flow problem, are
among the most foundational problems in combinatorial optimization
\cite{ahuja1993network}. Since their origins in the transportation problem \cite{schrijver2002history}, they have
served as a testbed for a variety of algorithmic paradigms 
\cite{schrijver2002history}; they are the central primitive in a large number of graph
algorithms \cite{li2025deterministic, li2021vertex, doi:10.1137/1.9781611977912.60, AbboudLPS23, AbboudKLPGSYY25}; and they arise directly in applications across transportation and
logistics \cite{ahuja1993network}, image segmentation
\cite{boykov2004experimental}, and machine learning, where computing optimal
transport distances reduces to a minimum-cost flow computation
\cite{pele2009fast, peyre2019computational}. Formally, we consider the following
problem.

\begin{definition}[Minimum-cost flow]\label{def:mincostFlow}
  Let $G = (V, E)$ be a directed graph with $n = |V|$ vertices and $m = |E|$ edges,
  let $\BB \in \R^{V \times E}$ denote its vertex-edge incidence matrix, and let
  $\uu \in \Z_{\geq 0}^{E}$ and $\cc \in \Z^{E}$ be integral edge capacities and
  costs. Given a demand vector $\dd \in \Z^{V}$ with $\vecone^\top \dd = 0$, the
  minimum-cost flow problem asks for a flow $\ff$ that achieves
\begin{equation}
\min_{\substack{\BB \ff = \dd \\ \veczero \leq \ff \leq \uu}} \cc^\top \ff
\label{eq:mincostFlowLP}
\end{equation}
  where $\BB\ff = \dd$ succinctly encodes that flow $\ff$ routes demand $\dd$ (flow conservation) and $0 \leq \ff \leq \uu$ encodes that the flow does not exceed capacities (capacity constraint).
\end{definition}

Throughout, we write $U = \|\uu\|_\infty$ and $C = \|\cc\|_\infty$, and we measure
running times in terms of $m$, $n$, $U$ and $C$. Since the constraint matrix in
Definition~\ref{def:mincostFlow} is totally unimodular, the optimum is attained at an integral
flow whenever the instance is feasible. 

\paragraph{Recent Progress on the Minimum-Cost Flow Problem.}
Decades of work produced a rich family of algorithms for minimum-cost flow, culminating in weakly polynomial cost scaling methods
\cite{goldberg1990finding} and in strongly polynomial algorithms
\cite{orlin1993faster}. By the 1990s several algorithms had been refined to the point
of being genuinely implementable and were being run on real instances at scale
\cite{bunnagel1998efficient, lobel1996solving}; the 1991 DIMACS Implementation
Challenge on network flows both reflected and accelerated this effort
\cite{johnson1993network}. 

Despite this sustained attention, an algorithm that solves the minimum-cost flow problem in provably almost-linear time remained elusive. By the end of the 1990s, the best runtimes remained $\tilde{O}(m^{3/2})$ for maximum flow \cite{goldberg1998beyond}, and $\tilde{O}(mn)$ for minimum-cost flow \cite{goldberg1990finding}. 

The decisive shift came in the early 2000s when Spielman and Teng \cite{spielman2004nearly} employed methods from convex optimization to solve electric flow problems in near-linear time. Building on this, Daitch and Spielman \cite{daitch2008faster} applied
Interior Point Methods (IPMs) to minimum-cost flow and obtained a running time of
$\widetilde{O}(m^{3/2} \log^{2} U)$, the first improvement on the problem to come out
of the continuous paradigm. 

Finally, in 2022, Chen, Kyng, Liu, Peng, Probst Gutenberg, and Sachdeva
\cite{chen2022maximum, chen2025maximum} achieved the first algorithm to solve minimum-cost flow exactly in
$m^{1 + o(1)} \log U \log C$ time -- a major milestone in algorithmic graph theory. Their algorithm combines a new interior point method, based on the work of Karmarkar \cite{karmarkar1984new} and Wallacher-Zimmermann \cite{WZ92}, with sophisticated graph data structures. The result has reshaped the
area, and has since been strengthened and simplified in several directions, including
deterministic and more modular variants
\cite{chen2024almost, van2023deterministic, van2024almost}.

\paragraph{The Theory-Practice Gap.} None of the IPM-based approaches, the centerpiece of algorithmic research since the 2000s, has so far reached practice. In particular, even the almost-linear-time algorithm based on the novel combinatorial IPM has not been implemented, and the practical state of the art relies entirely on insights preceding this new era. The reference point remains the experimental evaluation of
Kov\'acs \cite{kovacs2015minimum}, which benchmarks the minimum-cost flow solvers of
the LEMON library \cite{dezso2011lemon}. Two of them are usually the most efficient: 1) a \emph{cost scaling} solver, building on and refining the implementation of 
\cite{goldberg1997efficient} of the push-relabel cost scaling method of Goldberg and
Tarjan \cite{goldberg1990finding}, and 2) a \emph{network simplex} solver adapted to
minimum-cost flow \cite{ahuja1993network}. Network simplex is particularly fast on small networks of up to a few thousand nodes. Cost scaling scales better with instance size and takes over on larger inputs. Nonetheless, on most benchmarks, both solvers exhibit super-linear scaling \cite{kovacs2015minimum}.

There is therefore a substantial gap between the asymptotics of the best algorithms
we know and the asymptotics of the algorithms we can run. Closing it is one of the
stated motivations for the ongoing DIMACS Flow Challenge 2.0
\cite{dimacsFlowChallenge2}, the first DIMACS challenge devoted to flow problems
since \cite{johnson1993network}.
There is a concrete obstacle: the $m^{o(1)}$ factor in \cite{chen2025maximum} hides
recursively constructed data structures whose depth and constants make a faithful
implementation hopeless, so it is not clear if the IPM framework and sophisticated data structures could translate to a practical implementation.

\paragraph{Our Contribution.}
This paper gives the first practical implementation of the almost-linear-time
minimum-cost flow algorithm of \cite{chen2025maximum}. Our contributions are as
follows:
\begin{itemize}
    \item We give a practical implementation of the combinatorial IPM of \cite{chen2025maximum}. In doing so, we overcome obstacles arising from numerical precision by discretizing the IPM implementation. Further, our implementation removes the need for oracle access to the optimal cost. 
    
    \item We estimate progress more carefully, which allows us to take larger steps, thereby improving progress-per-step by orders of magnitude. We also introduce rules to drive the IPM more aggressively, making the IPM feasible in practice.

    \item We implement a heuristic version of the intricate min-ratio cycle data structure of  \cite{chen2025maximum} that repeatedly extracts progress steps. In particular, we start with a single shortest-path tree. We then iterate over its fundamental cycles to take IPM steps. After processing each of these cycles, we consider swapping the off-tree edge with a tree edge, which ensures good progress of later IPM steps along fundamental cycles in practice.
    
    \item We evaluate our solver against the LEMON cost scaling and network simplex solvers on 12 instance families from \cite{kovacs2015minimum}. We present the insights generated from our implementation and where future work should continue.
\end{itemize}

Despite larger constant factors of our solver, the results are already encouraging: our solver exhibits better empirical scaling than both LEMON solvers on instances from  \texttt{vision\_rnd} and \texttt{grid\_wide}. On the latter, we observe polynomial scaling of $m^{1.04}$ versus LEMON's $m^{1.48}$. 
Given how heavily engineered the LEMON solvers are, and that our implementation is a first attempt, which has not been optimized for constants, we read this as evidence that the almost-linear-time framework has at least some relevance in practice.

\section{Almost-Linear-Time Min-Cost Flow in Theory}
\label{sec:theory}

This section briefly introduces the theoretical foundation underlying our implementation, drawing on recent work on almost-linear-time algorithms for minimum-cost flow \cite{chen2025maximum, van2023deterministic, chen2024almost, van2024almost}. These algorithms use an interior point method as the outer loop, and show that the sequence of subproblems arising at each step can be solved efficiently using sophisticated graph data structures.

\paragraph{IPM.} 
The starting point of \cite{chen2025maximum, van2023deterministic, chen2024almost, van2024almost} is a potential-reducing IPM due to Karmarkar \cite{karmarkar1984new} and Wallacher-Zimmermann \cite{WZ92} who named it a \emph{combinatorial IPM}.
At each iteration, the method decreases the potential of the current solution $\ff$, which measures the solution's closeness to the optimal cost $C^*$ and its feasibility through a barrier term:
\begin{equation}
\Phi(\ff) = 20 m \log(\cc^{\top}\ff - C^*) + \mathcal{B}(\ff) \qquad \mathcal{B}(\ff) := \sum_{e \in E} \left( (\uu(e) - \ff(e))^{-\gamma} + (\ff(e))^{-\gamma}\right),
\label{eq:IPMPotential}
\end{equation}
where $\gamma = 1/(1000 \log mU)$. If $C^*$ is unknown, it can be determined by a binary search. The barrier $x^{-\gamma}$ plays the role of the more standard $-\log x$, with the precise form chosen for numerical stability.

The IPM requires an initial flow with $\BB\ff = \dd$ and $\veczero < \ff < \uu$, which we call \emph{strictly} feasible. Such a flow can always be constructed without changing the optimal solution \cite{chen2025maximum}.

At each step, the IPM updates the current solution by a circulation $\bDelta$ (i.e., $\BB\bDelta = \veczero$), setting $\ff' = \ff + \bDelta$. Flow conservation is maintained by induction: $\BB\ff' = \BB\ff = \dd$. The step size is restricted so that strict feasibility is preserved: $\veczero < \ff' < \uu$. To reduce the potential, $\bDelta$ is chosen to minimize the inner product with the gradient $\gg$ of $\Phi(\ff)$, relative to an $\ell_1$-norm weighted by the \emph{lengths} $\bell \in \mathbb{R}^E$ (arising from the second derivative). Writing $\LL$ for the diagonal matrix of $\bell$, the update problem is:
\begin{equation}
\min_{\BB \bDelta = \veczero} \frac{\gg^\top \bDelta}{\|\LL \bDelta\|_1}.
\label{eq:flowMinRatio}
\end{equation}
By a standard path-cycle decomposition, there always exists a cycle such that the unit circulation along the cycle is optimal for \eqref{eq:flowMinRatio}. We call this the \emph{min-ratio cycle} problem. It is shown in \cite{chen2025maximum} that augmenting along an $m^{o(1)}$-approximate min-ratio cycle decreases the potential by $\Omega(1/m^{o(1)})$, so after $m^{1+o(1)}$ iterations the IPM converges to a solution that can be rounded to the minimum-cost flow. The IPM further exhibits a certain robustness in that the multiplicative changes to $\gg$ and $\bell$ are bounded over the entire update sequence. Thus approximations of $\gg$ and $\bell$ can be maintained efficiently without requiring too many changes. 

\paragraph{Min-Ratio Cycle Data Structure.}
The central object in \cite{chen2025maximum} is an $m^{o(1)}$-approximate low-stretch spanning tree (LSST): a probabilistic spanning tree $T \subseteq G$ in which the $st$-distance under $\bell$ is at most $m^{o(1)}$ times the corresponding distance in $G$ in expectation, for all $s, t \in V$. By linearity, some \emph{fundamental cycle}, formed by an off-tree edge $e = (u,v)$ and the $u$-to-$v$ path in $T$, is an $m^{o(1)}$-approximate min-ratio cycle in expectation. Markov's inequality then gives that with probability at least $1/2$, at least one fundamental cycle is a $2m^{o(1)}$-approximation. Maintaining $\tau = O(\log n)$ independent LSSTs $T_1, \ldots, T_\tau$ and applying Chernoff's bound ensures that some $T_i$ contains a good fundamental cycle with high probability.

As flow is augmented and lengths $\bell$ are updated, the trees may degrade. The data structure in \cite{chen2025maximum} handles this by identifying off-tree edges $e$ and trees $T_i$ such that either the fundamental cycle of $e$ in $T_i$ is an approximate min-ratio cycle or swapping $e$ with an edge in $T_i$ restores its quality.

\section{Transfer to Practice}
\paragraph{Numerical Precision.} In our implementation of the IPM proposed in \cite{chen2025maximum}, we have to be particularly careful with numerical stability. Floating point numbers follow the idea of scientific notation with a fixed number of digits and thus the smallest representable unit depends on the magnitude of the number itself \cite{IEEEFloat}. The associated loss of associativity might result in breaking flow conservation when routing flow mass. 
As preserving flow conservation is highly critical, we restrict ourselves to routing small \emph{atoms} whose amount of flow $\alpha$ is guaranteed to be representable exactly by all flow amounts in the feasible range. Then, associativity holds again as addition and subtraction stay exact. As a side effect, as soon as $\cc^\top\ff < C^* + 1$, we can employ rounding from \cite{Cohen1995} to obtain an optimal integral solution, which is representable no matter the choice of $\alpha$. Smaller $\alpha$ allow us to control proximity to barriers at a finer granularity at the cost of more updates and thus runtime. Consequently, we choose the largest $\alpha$ that still leads to the exact optimum.

\paragraph{Loss of Strict Feasibility.} The idea of strict feasibility is to always allow us to route a small amount along every cycle in both directions. Using atoms to discretize the values that the flow can take necessarily leads to a loss of this property, as the flow is restricted to the feasible range and thus there must be some discrete threshold at which flow mass can only be routed in one direction. All that remains to choose is whether this threshold at which edges turn \emph{directed} is 0 or $\alpha$ away from the hard bounds of the feasible range. We choose it at 0 to make the exact optimum reachable by our IPM.

\paragraph{Refined IPM.}
The potential in (\ref{eq:IPMPotential}) depends on the cost of an optimal flow $C^*$, which we usually do not know. While it can easily be binary-searched, this adds an undesirable logarithmic factor in practice. 

Instead, we can introduce a \emph{cost scale} $\lambda$ that does not depend on $C^*$. This cost scale is initialized such that the cost and barrier term have the same magnitude. We use barrier $\mathcal{B}'(\ff)$ obtained by shifting $\mathcal{B}(\ff)$ horizontally by a single atom $\alpha$ to ensure well-definedness even under the loss of strict feasibility. To ease scaling, we also shift the barrier vertically by $\mathcal{B}'(\uu/2)$, which yields that the minimum of the barrier terms $ \mathcal{B}'(\ff) -  \mathcal{B}'(\uu/2)$ becomes $0$. We summarize the potential used in our implementation below:
\begin{equation}
\Phi'(\ff) = \lambda \cc^\top \ff + \mathcal{B}'(\ff) -  \mathcal{B}'(\uu/2) \qquad \mathcal{B}'(\ff) := \sum_{e \in E} \left( (\uu(e) - \ff(e) + \alpha)^{-\gamma} + (\ff(e) + \alpha)^{-\gamma} \right).    \label{eq:NewIPMPotential}
\end{equation}
We now periodically increase $\lambda$,  crudely mimicking the effect of the logarithm in (\ref{eq:IPMPotential}). We give more details on how to increase $\lambda$ in the last paragraph of this section.

\paragraph{Refined Update Problem.}
As seen in Section \ref{sec:theory}, \cite{chen2025maximum} routes a small amount on an approximate min-ratio cycle to make good progress and guarantee almost-linearly many updates. While the ratio between gradient and length of a cycle serves as a good proxy for progress in the IPM, our implementation makes more accurate estimations and takes more carefully calculated steps. More precisely, the progress made by pushing an update amount $\delta > 0$ along an oriented cycle consists of three summands: 1) the linear change in cost (total cost along the cycle times $\delta$); 2) the convex change in \emph{forward} barriers (increasing in $\delta$); 3) the convex change in \emph{reverse} barriers (decreasing in $\delta$). This split allows us to bound each term individually, obtain a better estimate and take much larger steps. We call the cycles that optimize our function \emph{max-progress} cycles.

\paragraph{Bounding the Progress.} Since the progress is the sum of the three convex functions, flow mass has to be routed in the direction of negative gradient to reduce the potential. We bound the summands as follows: 1) for the cost we keep the exact term; 2) for the increase in forward barrier terms we use a second-order approximation with the second-order term scaled depending on the ratio between $\delta$ and the minimum residual capacity of any edge such that this yields a sound upper bound; 3) for the decrease in reverse barrier terms we use a single barrier term which we offset and scale such that the first and second derivatives match the respective sums of the actual reverse barriers (proof in Appendix \ref{appendix:proof_of_upper_bound}).

\paragraph{Computing Step Size and Progress.} Because this upper bound is convex again, we can find its minimum efficiently and an approximate minimum, which suffices in practice, even more efficiently. In total, computing the progress bound requires six aggregates of the cycle: the total cost, the total first derivative and second derivative of the forward and reverse barriers, and the minimum residual capacity of the forward edges.

\paragraph{Finding Max-Progress Cycles.} In \cite{chen2025maximum}, the min-ratio cycle data structure maintains $m^{o(1)}$-approximate LSSTs $T_1, T_2, \ldots, T_\tau$ for $\tau = \mathcal{O}(\log n)$ along with estimates of the ratio for all fundamental tree cycles. However, these LSSTs are intricate and complex, and thus not feasible in practice, leading us to a simplified approach: we only maintain a single tree $T$, and we do not enforce it to be an LSST of good quality. Instead, we fix a root $r$, and initialize $T$ to be a single-source shortest paths (SSSP) tree of $r$ with respect to the lengths $\bell$. We observed empirically that $T$ is typically a good LSST.  We have shared this insight with a group of researchers who have since shown that shortest path trees yield an $\mathcal{O}(\log^5 n)$-approximate LSST under smoothed analysis \cite{heuristicLSST}. 

As in \cite{chen2025maximum}, we evaluate all fundamental cycles and calculate their progress. To collect the required aggregates efficiently, we use a depth-first search on the tree along with a modified union-find data structure \cite{tarjan1983data}. We then choose the max-progress cycles to update our solution. Since our progress estimates are more informative than the gradient-length ratios, this allows us to find a fundamental cycle that makes at least as much progress as the min-ratio cycles of our tree $T$. For efficiency, we maintain $T$ in a link-cut data structure (like \cite{chen2025maximum}), which enables fast queries given an off-tree edge for the progress by maintaining an $\varepsilon$-approximation of the aggregates under fast flow updates along fundamental cycles. 

\paragraph{Maintaining Max-Progress Cycles.} To find good fundamental cycles, we maintain a queue that groups the off-tree edges by the magnitude of the progress achieved by their fundamental cycles. We then repeatedly pop the edge with the largest progress estimate in the queue and use it for an IPM step. Note that the progress initially computed for off-tree edges might become stale, i.e., change significantly. Thus, whenever an edge is popped from the queue, we first re-evaluate the aggregates and re-calculate the optimal step size.

As lengths $\bell$ change, we observed that the LSST quality deteriorates quickly, leading to low progress. Recomputation from scratch is, however, too costly, as our IPM runs for many iterations. Instead, we also repair our tree $T$. To this end, after considering a fundamental cycle to update the solution, we check whether swapping its off-tree edge brings the tree closer to an SSSP tree again, and thus to an LSST. In practice, these repairs are so efficient that we never have to redraw the LSST from scratch. 

Lastly, we say that inspecting $2m$ off-tree edges constitutes a \emph{phase} and at the beginning of a new phase, we recompute all aggregates exactly and repopulate our queue with fresh estimates. We point out two remarkable facts. 1) The SSSP tree seems to always allow for good updates, suggesting that the IPM lengths are sufficiently random for the analysis in \cite{heuristicLSST} to hold, 2) we generally do not need to converge closely to the exact SSSP; few updates typically suffice. 

\paragraph{Dealing With Directed Edges.}
Recall how atomic flow leads to directed edges, and we have so far ignored this issue in our description of the data structure. Having directed edges in the tree $T$ is unfavorable because all cycles with opposite directed edges become \emph{unroutable}. To mitigate this, we enforce as an invariant that we can always route flow away from the root. We found this to reduce the number of unroutable cycles drastically in practice. However, we can only maintain this invariant if all nodes are reachable from the root, which is not true in general graphs. Thus, we split the residual graph of the initial flow into strongly connected components during preprocessing and solve them individually: clearly, this enables the invariant while leaving the optimum unchanged (it is not possible to route any flow on edges between components).

\paragraph{Adjusting the Cost Scale.} We previously defined a \emph{phase} to be the inspection of $2m$ off-tree edges. Recall that at the beginning of each phase, we recompute all progress estimates for off-tree edges. Note that due to atomic flow, this estimate might be 0 because routing $\alpha$ in the direction of negative gradient can already overshoot the minimum of our bound. We take a high share of such cycles as our signal that the IPM is well-balanced and thus double the cost scale $\lambda$ until at least a constant fraction (10 \%) of cycles makes progress again. Eventually, when almost all cycles either have an edge directed against the direction of negative gradient or have cost 0, doubling $\lambda$ will have almost no effect anymore and we instead round to an integral solution and terminate. Otherwise, we compute the $\varepsilon$ to use for this phase. Larger $\varepsilon$ reduce the work required to maintain the link-cut tree but worsen our progress estimates (because we have to assume that we are on the worst-case end of the $\varepsilon$-band around the current aggregates). In practice, choosing $\varepsilon$ such that at least half of the cycles that make progress at $\varepsilon=1$ maintain at least 65 \% of their progress works well. We also employ several heuristics to detect livelocks due to numerical discretization.

\section{Experiments}
\paragraph{IPM Implementation.}
Our IPM chooses $\alpha$ as a compile-time constant power of two. Further, we choose $\gamma=2$: larger $\gamma$ tend to admit exact convergence under larger $\alpha$ but past a certain threshold, we no longer converge to the exact optimum.

\paragraph{Baselines and Instances.} We follow the call for papers and compare our IPM implementation to the network simplex and cost scaling implementations from LEMON 1.3.1 \cite{dezso2011lemon} on the core instances from \cite{kovacs2015minimum}: \texttt{netgen\_*} \cite{klingman1974netgen}, \texttt{goto\_*} \cite{goldberg1992implementing}, \texttt{grid\_*} \cite{resende1993interior}, \texttt{road\_*} \cite{demetrescu2009shortest}, and \texttt{vision\_*} \cite{uwomaxflow,Boykov2006}.
We also extend \texttt{grid\_wide}, for which our solver showed promising scaling, by 6 more sizes (with a modified GRIDGRAPH).

\paragraph{Machine and Setup.}
We conduct our experiments on a machine with an AMD EPYC 7742 processor with 64 cores (hyperthreading disabled). The L1 instruction cache is 32 KiB per core and data caches are 32 KiB per-core L1, 512 KiB per-core L2, and 16 MiB per-4-cores L3. The CPU has a base frequency of 2.25 GHz and has frequency scaling enabled (up to 3.4 GHz; cannot be disabled due to cloud setting). The machine has 512 GB of DDR4 RAM. We run Ubuntu 24.04.5 LTS with kernel 6.8.0-139-generic. All code was compiled using g++ 13.3.0 and flags \texttt{-std=c++23 -O3 -march=native -ffast-math -fno-finite-math-only}. Due to time constraints, we only measured the wall time of a single run on a single seed for each family/size pair. Only one solver was active at a time and no other work was performed simultaneously. The first instance of each family is started at $\alpha=1$, which is then halved until we reach the exact optimum. Each consecutive point follows the same strategy warm-started from the $\alpha$ of the previous instance. We assume oracle access to the largest $\alpha$ that achieves the exact minimum for our runtime plots.

\paragraph{Plots.} On page \ref{page:overview} we present our evaluation. On the x-axis, we report the size of the instance in edges and on the y-axis the runtime in seconds. Both axes are logarithmic. Each line represents one solver: orange is LEMON's network simplex, purple is LEMON's cost scaling and blue is  our IPM implementation. The small numbers above the IPM line mark increases in bits required to solve the instance exactly ($\alpha = 2^{-\text{bits}}$). We report timeouts after 2 hours (red line) with a cross. Timeouts did not apply to the extended \texttt{grid\_wide} instances but we retired LEMON cost scaling after it was persistently slower than our IPM.

\paragraph{Failures.} Two IPM points marked with a red star returned a feasible but non-optimal flow, for distinct reasons. On \texttt{grid\_long}, our termination rule fires too early (smaller $\alpha$ cannot compensate). This issue persists on the next larger instance of the family, which is why we retire the IPM on this family. On  the largest \texttt{road\_flow} instance, the starred point is the smallest $\alpha$ at which the IPM terminates at all; at every smaller $\alpha$ we enter a livelock that our current heuristics cannot detect. 

\paragraph{Scaling.}
For the remaining families, we mainly care about the exponent of the least-squares power-law fit (which is the slope in the log-log plot) because our implementation is not highly optimized and thus loses a constant factor compared to the baselines. We highlight a few selected exponents in the plots. We observe that our exponent is competitive with the best baseline solver for many families: \texttt{netgen\_sr}, \texttt{netgen\_deg}, \texttt{grid\_square}, \texttt{road\_paths}, and \texttt{vision\_prop}. For two families, our solver wins: slightly for \texttt{vision\_rnd} and significantly for \texttt{grid\_wide} ($m^{1.04}$ vs. $m^{1.48}$). In contrast, we lose on \texttt{netgen\_8} ($m^{1.41}$ vs. $m^{1.23}$) and \texttt{goto\_sr} ($m^{1.42}$ vs. $m^{1.27}$). For \texttt{goto\_8}, we encounter a lot of variance on all solvers.

\setlength{\intextsep}{0pt}
\begin{wrapfigure}{r}{0.5\textwidth}
    \centering
    \includegraphics[
        width=0.48\textwidth,
        keepaspectratio,
        clip
    ]{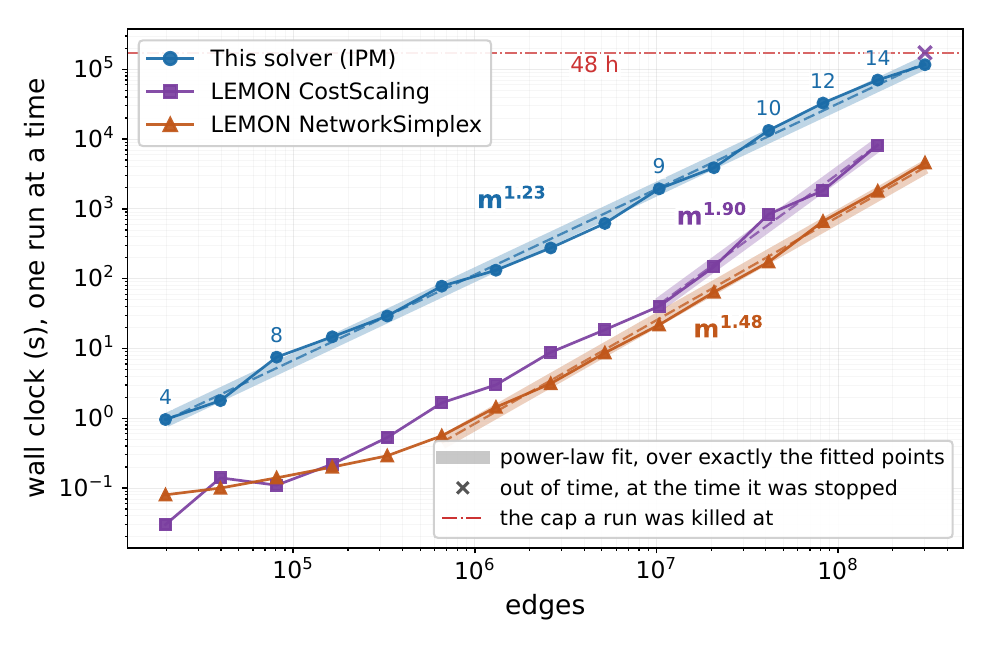}
    \caption{An extended scaling experiment on the \texttt{LB07-bunny-lrg} surface-fitting instance.}
    \label{fig:bunny}
\end{wrapfigure}
\paragraph{Extending \texttt{vision\_rnd}.}
Because the exponent in the overview figure seemed promising, we conducted a scaling experiment on a family constructed from the \texttt{LB07-bunny-lrg} surface-fitting instance \cite{uwomaxflow,lempitsky2007global}. Similarly to \texttt{vision\_rnd} in the overview, our IPM has the best exponent. Further, for the largest datapoint, we beat LEMON cost scaling, which times out at 48h. However, LEMON network simplex shows a much lower exponent than in the overview, which keeps our IPM from crossing below both baselines. Nevertheless, we observe stable scaling with a low exponent over more than 4 decades. 

\section{Insights}
\paragraph{Other Approaches.}
Before settling on the design introduced above, we also evaluated several alternatives: extending the barrier past the feasible range to avoid directed edges, different LSST designs, different rules to increase the cost scale and detect convergence, using forests instead of trees to route flow in a parallel and cache-friendly manner, and scaling precision dynamically instead of choosing statically.
While some of these approaches showed interesting constant improvements, we decided to submit the presented algorithm because it is more stable and scales competitively across many instance families and sizes.

\paragraph{Coupling Between Cost Scale and Precision.}
One particularly interesting lesson of our research is that $\lambda$ and $\alpha$ are coupled. Increasing $\lambda$ very carefully leads to convergence to the optimum even at large $\alpha$ but also bad runtime scaling. On the other hand, we came up with a formula for a constant choice of $\lambda$ (and a corresponding $\alpha$) for which the IPM always converges to the exact optimum but first experiments have shown bad scaling too. As the presented algorithm exhibits better scaling than either extreme, choosing a good strategy to drive the cost scale and adequate precision seems crucial for good scaling.

\paragraph{Future Work.}
Our current rules seem to work especially well for the \texttt{grid\_wide} family, demonstrating that our approach can scale (almost) linearly in practice. Thus, understanding why our IPM scales well on this family and struggles on others could lead to even better steering rules for the IPM and practical almost-linear minimum-cost flow on more, if not all, instance families.
\overviewpage

\printbibliography 

\clearpage
\appendix
\section{AI Use}
We used AI to spell and grammar check, get feedback on readability, polish individual sentences, generate the plots, and find related literature. We used agentic AI for coding and evaluation for some parts of the implementation. The proof in Appendix \ref{appendix:proof_of_upper_bound} was generated with AI; two authors independently verified it.
\section{Proof of Upper Bound}
\label{appendix:proof_of_upper_bound}
\noindent\textbf{Setting.}
Fix a real number $\gamma>0$. For $a>0$, $b>0$ let
\[
  f_{a,b}(x) := a(b+x)^{-\gamma} - a b^{-\gamma}, \qquad x\ge 0,
\]
so that $f_{a,b}(0)=0$ and
\[
  f_{a,b}'(x) = -\gamma\, a\,(b+x)^{-\gamma-1}, \qquad
  f_{a,b}''(x) = -\gamma(-\gamma-1)\, a\,(b+x)^{-\gamma-2}.
\]
Given $n\ge 1$, $a_1,\dots,a_n>0$ and $b_1,\dots,b_n>0$, put
$S := \sum_{i=1}^n f_{a_i,b_i}$; note $S(0)=0$.
Since a sum of terms $(b_i+x)^{-\gamma}$ with distinct $b_i$ is not itself of the
form $a(b+x)^{-\gamma}$, the conditions $f'=S'$ and $f''=S''$ are imposed at the
single point $x=0$.

\begin{theorem}[Bernoulli's inequality, negative exponent]\label{lem:bernoulli}
Let $r<0$. Then $(1+u)^{r}\ \ge\ 1+ru$ for all $u>-1$, with equality if and only if $u=0$.
\end{theorem}

\begin{proof}
Let $\phi(u) := (1+u)^{r}-1-ru$ for $u>-1$. Then $\phi(0)=0$ and
\[
  \phi'(u) = r\big((1+u)^{r-1}-1\big).
\]
Because $r-1<0$, the map $t\mapsto t^{\,r-1}$ is strictly decreasing on $(0,\infty)$,
so $(1+u)^{r-1}-1$ is positive for $-1<u<0$ and negative for $u>0$. Since $r<0$,
we get $\phi'<0$ on $(-1,0)$ and $\phi'>0$ on $(0,\infty)$. Hence $\phi$ is
strictly decreasing on $(-1,0]$ and strictly increasing on $[0,\infty)$, so
$\phi\ge\phi(0)=0$ with equality only at $u=0$.
\end{proof}

\begin{theorem}\label{thm:main}
Let $\gamma>0$, $a_1,\dots,a_n>0$, $b_1,\dots,b_n>0$ and $S=\sum_{i=1}^n f_{a_i,b_i}$.
There exist unique $a>0$, $b>0$ with
\[
  f_{a,b}'(0)=S'(0) \qquad\text{and}\qquad f_{a,b}''(0)=S''(0),
\]
namely
\[
  b=\frac{\sum_{i=1}^n a_i b_i^{-\gamma-1}}{\sum_{i=1}^n a_i b_i^{-\gamma-2}},
  \qquad
  a=b^{\,1+\gamma}\sum_{i=1}^n a_i b_i^{-\gamma-1}.
\]
For this choice,
\[
  f_{a,b}(x)\ \ge\ S(x) \qquad\text{for all } x\ge 0,
\]
and the inequality is strict for every $x>0$ unless $b_1=\dots=b_n$
(in which case $f_{a,b}=S$).
\end{theorem}

\begin{proof}
Write $f:=f_{a,b}$ and
\[
  M:=\sum_{i=1}^n a_i b_i^{-\gamma-1}>0, \qquad
  W:=\sum_{i=1}^n a_i b_i^{-\gamma-2}>0 .
\]

\emph{Step 1: the fit.}
Dividing $f'(0)=S'(0)$ by $-\gamma\neq 0$ and $f''(0)=S''(0)$ by
$-\gamma(-\gamma-1)\neq 0$, the two conditions read
\[
  a b^{-\gamma-1}=M, \qquad a b^{-\gamma-2}=W .
\]
Dividing the first equation by the second gives $b=M/W$, and then the first
gives $a=Mb^{\,1+\gamma}$; both are positive, and these values satisfy both
equations. This proves existence, uniqueness and the stated formulas.

\emph{Step 2: rewriting $S'$ and $f'$.}
Put $w_i:=a_i b_i^{-\gamma-2}>0$. Then
\[
  W=\sum_{i=1}^n w_i, \qquad M=\sum_{i=1}^n w_i b_i, \qquad
  a_i b_i^{-\gamma-1}=w_i b_i ,
\]
so $b=M/W$ is the $w$-weighted mean of $b_1,\dots,b_n$.
Fix $x\ge 0$ and set
\[
  r:=-\gamma-1<0, \qquad
  t_i:=1+\frac{x}{b_i}>0, \qquad
  v_i:=\frac{w_i b_i}{M}>0, \qquad\text{so that}\qquad \sum_{i=1}^n v_i=1 .
\]
Then
\[
  1+\frac{x}{b} = 1+\frac{xW}{M}
  = \frac{1}{M}\sum_{i=1}^n \big(w_i b_i + w_i x\big)
  = \sum_{i=1}^n \frac{w_i b_i}{M}\Big(1+\frac{x}{b_i}\Big)
  = \sum_{i=1}^n v_i t_i ,
\]
and therefore, using $a_i b_i^{-\gamma-1}=Mv_i$ and $a b^{-\gamma-1}=M$,
\begin{align*}
  S'(x) &= -\gamma\sum_{i=1}^n a_i (b_i+x)^{r}
         = -\gamma\sum_{i=1}^n a_i b_i^{-\gamma-1}\Big(1+\frac{x}{b_i}\Big)^{r}
         = -\gamma M\sum_{i=1}^n v_i\, t_i^{\,r}, \\
  f'(x) &= -\gamma\, a\,(b+x)^{r}
         = -\gamma\, a b^{-\gamma-1}\Big(1+\frac{x}{b}\Big)^{r}
         = -\gamma M\Big(\sum_{i=1}^n v_i t_i\Big)^{r} .
\end{align*}

\emph{Step 3: Jensen's inequality for $t\mapsto t^{\,r}$.}
Let $m:=\sum_{i=1}^n v_i t_i>0$ and $u_i:=t_i/m-1>-1$. Then
$\sum_{i=1}^n v_i u_i=\frac{1}{m}\sum_{i=1}^n v_i t_i-1=0$, and
Theorem~\ref{lem:bernoulli} gives
\[
  \sum_{i=1}^n v_i\, t_i^{\,r}
  = m^{r}\sum_{i=1}^n v_i (1+u_i)^{r}
  \ \ge\ m^{r}\sum_{i=1}^n v_i (1+r u_i)
  = m^{r}
  = \Big(\sum_{i=1}^n v_i t_i\Big)^{r}.
\]
Since all $v_i>0$, equality holds if and only if $u_1=\dots=u_n=0$, i.e.
$t_1=\dots=t_n$, i.e. $x=0$ or $b_1=\dots=b_n$.

\emph{Step 4: conclusion.}
Multiplying the inequality of Step~3 by $-\gamma M<0$ and using Step~2 yields
\[
  S'(x)\ \le\ f'(x) \qquad\text{for all } x\ge 0 .
\]
Since $f(0)=0=S(0)$,
\[
  f(x)-S(x)=\int_0^x \big(f'(t)-S'(t)\big)\,dt\ \ge\ 0 \qquad (x\ge 0).
\]
If the $b_i$ are not all equal, the equality statement of Step~3 shows that the
continuous integrand is strictly positive on $(0,x]$, hence $f(x)>S(x)$ for
$x>0$. If $b_1=\dots=b_n=:\beta$, then $M=\beta^{-\gamma-1}\sum_i a_i$ and
$W=\beta^{-\gamma-2}\sum_i a_i$, so the fit returns $b=\beta$ and $a=\sum_i a_i$,
and $f=S$.
\end{proof}

\begin{remark}
The argument uses only the power rule $\frac{d}{dt}t^{\,s}=s\,t^{\,s-1}$ on
$(0,\infty)$ and Theorem~\ref{lem:bernoulli}, both valid for arbitrary real
exponents, so $\gamma$ may be any positive real number.
\end{remark}
\end{document}

%% file: defs.tex
\usepackage{enumitem}
\usepackage{graphicx}
\usepackage{fancybox}
\usepackage{comment}
\usepackage{xcolor}
\usepackage{nameref}
\usepackage{bbm}

\definecolor{ForestGreen}{rgb}{0.1333,0.5451,0.1333}
\definecolor{DarkRed}{rgb}{0.8,0,0}
\definecolor{Red}{rgb}{1,0,0}
\usepackage[linktocpage=true,
colorlinks,
linkcolor=ForestGreen,citecolor=ForestGreen,
bookmarks,bookmarksopen,bookmarksnumbered]
{hyperref}

\usepackage{amsmath}
\usepackage{amssymb}
\usepackage{amsthm}
\usepackage{mathtools}

\usepackage{subfig}

\usepackage[ruled, noend, linesnumbered]{algorithm2e}

\usepackage{float}
\usepackage{cleveref}

\newcommand\veczero{\boldsymbol{0}}
\newcommand\vecone{\boldsymbol{1}}

\newcommand\R{\mathbb{R}}

\newcommand\Z{\mathbb{Z}}

\newcommand{\bell}{\boldsymbol{\ell}}

\newcommand{\bDelta}{\boldsymbol{\Delta}}

\renewcommand{\tilde}{\widetilde}

\renewcommand{\epsilon}{\ensuremath\varepsilon}

\renewcommand{\phi}{\ensuremath{\varphi}}

\newcommand\BB{\boldsymbol{\mathrm{{B}}}}

\newcommand\LL{\boldsymbol{\mathrm{{L}}}}

\newcommand\cc{\boldsymbol{\mathrm{c}}}
\newcommand\dd{\boldsymbol{\mathrm{d}}}

\newcommand\ff{\boldsymbol{\mathrm{f}}}
\renewcommand\gg{\boldsymbol{\mathrm{g}}}

\newcommand\uu{\boldsymbol{\mathrm{u}}}

\SetKw{Break}{break}
\SetKwComment{Comment}{$\triangleright$\ }{}
\SetCommentSty{mycommfont}

%% file: theoremsetup.tex
\newtheorem{theorem}{Theorem}[section]

\newtheorem{definition}[theorem]{Definition}
\newtheorem{remark}[theorem]{Remark}